\documentclass[12pt]{amsart}
\usepackage[foot]{amsaddr}
\usepackage{amsmath,amsthm,amssymb}
\usepackage[margin=0.8in]{geometry}
\usepackage{mathtools}
\usepackage{ulem}
\usepackage{xcolor}
\usepackage{tikz}
\usepackage{xcolor}
\definecolor{dblue}{RGB}{20,66,129}
\usepackage{graphicx}
\graphicspath{{Figures/}}
\usetikzlibrary{positioning}
\usetikzlibrary{intersections}

\numberwithin{equation}{section}

\newtheorem{theorem}{Theorem} 
\newtheorem{lemma}{Lemma}

\newtheorem{corollary}{Corollary}
\newtheorem{proposition}{Proposition}
\newcommand{\R}{\mathbb{R}} 
\newcommand{\Z}{\mathbb{Z}} 

\newcommand{\In}{{\rm{in}}}

\usepackage{hyperref}
\hypersetup{hidelinks}

\title{Competition in the Nutrient-Driven Self-Cycling Fermentation
Process}

\date{}

\author{Stacey R. Smith?$^1$ \and Tyler Meadows$^2$ \and Gail S.K. Wolkowicz$^3$}

\address[1]{Department
of Mathematics and Faculty of Medicine, University of Ottawa.
The work of this author was partially supported by the Natural Sciences and Engineering Research Council of Canada. This research is based on
part  of this author's Ph.D. thesis at McMaster University.}
\address[2]{Department of Mathematics, Queen's University}
\address[3]{Department of Mathematics and
Statistics, McMaster University.  
The work of this author was partially supported by the Natural Sciences and Engineering Research Council of Canada.}
\begin{document}

\begin{abstract} Self-cycling fermentation is an automated process used
for culturing microorganisms. We consider a model of $n$ distinct species competing
	for a single non-reproducing nutrient in a self-cycling fermentor in
	which the nutrient level is used as the decanting condition. The
model is formulated in terms of impulsive ordinary differential equations.
	We prove that two species are able to coexist in the fermentor under certain conditions. We also provide numerical simulations that
suggest coexistence of three species is possible and 
that competitor-mediated coexistence can occur in this case. These results are in contrast to the chemostat, the continuous analogue, where multiple species cannot coexist on a single nonreproducing nutrient.

\end{abstract}
\maketitle

\section{Introduction}
Self-cycling fermentation (SCF) is a technique used to culture microorganisms. In
this process, a tank is filled with a liquid medium that contains all the nutrients
required for microbial growth. The medium is inoculated with microorganisms that  use the nutrient to grow and reproduce. The contents of the tank are carefully monitored by a computer, and when predefined conditions (called the decanting criteria) are met, the computer then instigates a rapid
emptying and refilling process, called a decanting process. During the decanting process, a set fraction of the contents of the tank
is removed and replaced by an equal volume of fresh
medium. Once the fresh medium has been added to the tank, the process begins anew,
with the microorganism consuming the new medium until the decanting criteria are met again. Under the right
conditions, this cycling continues indefinitely, and the process does not
require an operator or any estimate of the natural cycle time of the
microorganisms in advance.

SCF was originally developed as a method to cultivate synchronized cultures of
bacteria; i.e., cultures in which all cells are the same age
\cite{BC1,sheppard1990development}. The process quickly found use in wastewater
treatment~\cite{HC,nguyen2000application,SaC}, where the decanting criteria could be
set so that the treated medium conformed to standards set by environmental-protection
agencies.  A two-stage variation on the SCF process has been used for bacteriophage
cultivation~\cite{sauvageau2010two}. Bacteriophages have been identified as useful
biomedical tools,  not only in the application of phage therapy
\cite{housby2009phage}, but also in bacterial control~\cite{kocharunchitt2009use} and the production of recombinant proteins for drug delivery~\cite{oh2007construction}. SCF has also shown promise as a method to produce some biologically derived compounds such as shikimic acid~\cite{tan2021transcriptomic}, which is an important component of the antiviral drug Oseltamivir, and cellulosic ethanol~\cite{wang2021co,wang2020improved}, which is a type of biofuel produced from otherwise unusable plant fibres. 

The original model of SCF was developed using the dissolved oxygen concentration as the decanting condition~\cite{WCR}. The nutrient-driven process, which uses a value of the nutrient concentration as the decanting condition, has been analyzed more thoroughly. The initial model of the nutrient-driven SCF process~\cite{SmWo2} was used to determine an optimal decanting fraction to maximize fermentor throughput under the assumption that the fermentor was being used for wastewater treatment. The nutrient-driven SCF model has been extended to investigate the role of cell size~\cite{SmWo1}, to investigate how resources that are inhibitory at high concentrations affect the process~\cite{fan2007analysis} and to investigate how multiple resources affect the long-term dynamics~\cite{hsu2019growth,meadows2020growth}. The model is described using impulsive differential equations, which accurately describe semi-continuous systems when the period being approximated is short compared to the cycle times~\cite{BS1,BS2}. In the case of self-cycling fermentation, the emptying and refilling process is fast compared to the time between such events, making it an ideal process for modelling with impulsive differential equations.

The outcomes of multiple-species competition have been discussed in many other scenarios. In the chemostat with constant input resource concentration and dilution rate, the species that can subsist on the lowest resource concentration will exclude all others~\cite{SW,wolkowicz1992global}. In contrast, an arbitrary number of species are able to coexist in the periodic chemostat, provided that certain conditions are met~\cite{wolkowicz1998n}. Similarly, at least two species have been shown o to coexist in serial transfer cultures~\cite{smith2011bacterial}, which can be thought of as a time-driven self-cycling fermentation process. Many of these theoretical results have also been verified experimentally~\cite{good2017dynamics,hansen1980single}. 

In wastewater systems, operators often want to curate an environment that selects for one species over another~\cite{dutta2015sequencing}. For example, one of the main challenges facing the full-scale implementation of anaerobic ammonium oxidation is the competition between nitrite-oxidizing bacteria and the desired anaerobic-ammonium-oxidizing bacteria~\cite{trinh2021recent}. Similarly, glycogen-accumulating organisms must be excluded from biological phosphate-removal systems since their presence can lead to reduced efficiency or even reactor failure~\cite{shen2016enhanced}. On the other hand, mixed-culture systems show promise as a method to reduce the production costs of some biologically manufactured plastics such as polyhydroxyalkanoates~\cite{nguyenhuynh2021insight}. Therefore, a solid theoretical understanding of the mechanisms that lead to the coexistence of multiple species or competitive exclusion is important in order to achieve desired outcomes. 


This paper is organised as follows. In Section \ref{nspecies}, we introduce
the model for $n$ species competing for a single limiting nutrient. In
Section \ref{twospecies}, we consider a simplified version of the model
with two species, run some numerical simulations that suggest
coexistence under certain conditions, and present our main theorem. We
prove that two species can coexist on a single nonreproducing nutrient,
under certain conditions. In Section \ref{threespecies}, we display
numerical simulations that suggest three species can also survive on a
single limiting nutrient, and we demonstrate that such survival is an
example of competitor-mediated coexistence. In Section \ref{discussion}, we discuss the
implications of the results. The proofs of all of the results
can be found in Appendix~\ref{proofs}. 


\section{A model for $n$ competing species}
\label{nspecies}
For a given function $z(t)$ and time $t$, let
$ z(t^-)\equiv \lim_{t\to t^-} z(t)$ and $z(t^+)\equiv\lim_{t\to t^+} z(t)$. We consider the following model for $n$ species competing for a single growth-limiting nutrient in a nutrient-driven self-cycling fermentor:
\begin{subequations}\label{competition}
\begin{align}
\begin{drcases}\begin{aligned}
\frac{ds}{dt} &= - \sum_{j=1}^n\frac{x_j f_j(s)}{Y_j} \\
\frac{dx_j}{dt} &= x_j(f_j(s)-d_j) \quad\qquad
j=1,\ldots, n\end{aligned} \end{drcases}
&& s(t)&\ne\overline{s}\label{competition_continuous} \\ \begin{drcases}\begin{aligned}
s(t^+) &= rs^\In+(1-r)s(t^-) \\
x_j(t^+) &= (1-r)x_j(t^-)  \ \quad\qquad j=1,\ldots, n 
\end{aligned} \end{drcases} && s(t^-) &= \overline{s}.
\end{align}
\end{subequations}
This model is a generalization of the model described by Smith and Wolkowicz~\cite{SmWo2}. 
Here, $s$ denotes the concentration of nutrient in the
fermentation vessel, $x_j$ is the
biomass of the $j$th population of microorganisms that consume the
nutrient, $Y_j$ is the cell yield constant, $d_j$ is the natural decay rate of the $j$th population,  $\bar s$ is the nutrient concentration that triggers the decanting process, $s^\In$ is the concentration of nutrient in the medium added during the decanting process and $r$ is the fraction of medium removed during the decanting process. We assume that $Y_j>0,\ \ s^\In>\bar{s} > 0,\ \ d_j \ge 0$ and $0<r<1$. We note that by rescaling $x_j$ by the factor $\frac{1}{Y_j}$,
these yield constants can be eliminated from the model. This rescaling is equivalent to setting each yield constant to 1. Thus, we consider this rescaled model for the remainder of the paper.

The functions $f_j:\mathbb{R} \to \mathbb{R}$ describe the rate at which the $j$th species consumes nutrient and converts it to biomass. We assume the $f_j$ are continuously differentiable, monotone non-decreasing and satisfy $f_j(0) = 0$. This class of functions includes the commonly used mass-action and Monod forms~\cite{pervez2018ecological}.
In numerical simulations,  we will use the Monod form for the response functions:
\begin{eqnarray*}
f_j(s) = \frac{m_j s}{K_j + s}, \qquad\qquad j \ = \ 1,\ldots,n,
\end{eqnarray*}
where $m_j$ is the maximum specific growth rate and $K_j$ is the half
saturation constant for the $j$th species. That is, $f_j(K_j) = \frac{1}{2}m_j.$

For each $j\in\{1,\dots,n \} $, let $\lambda_j$ denote the nutrient concentration at which
$f_j(\lambda_j)= d_j$. 
These values are referred to as {\it break-even concentrations}, since if the nutrient level were to be held constant at $\lambda_j$, then the $j$th species would not experience any growth or decay. 

Note that since $s$ is decreasing, if $s(0)<\bar s$, then
$\bar s$ 
is never reached and there will be no impulsive effect. In this case, the system will
approach an initial-condition-dependent equilibrium point with $s = 0$ or $s= s(0)$
if $x_j(0) = 0$ for all $j\in\{1,...,n\}$. We assume, without loss of
generality, that $s(0)> \bar{s}$, so that there is no immediate impulsive effect. For simplicity of notation, define 
\begin{align*}
\bar s^+ \equiv rs^\In+(1-r)\bar s.
\end{align*}

For each $j\in\{1,\dots,n \}$,  let
\begin{align*}
\mu_j &\equiv \int^{\bar{s}^+}_{\bar s}1-\frac{ d_j}{f_j(s)}ds.
\end{align*}
This represents the net growth in the $j$th species throughout one cycle when it is the only species present in the fermentation vessel. 

Throughout, we will make the technical assumption that 
\begin{equation}\label{eq:mu_min}
    \mu_{min} \equiv \int_{\bar{s}}^{\lambda_{max}} \frac{\min_j(f_j(s)-d_j)}{\min_j(f_j(s))}ds +\int_{\lambda_{max}}^{\bar{s}^+} \frac{\min_j(f_j(s)-d_j)}{\max_j(f_j(s))}ds > 0,
\end{equation}
where $\lambda_{max} = \max\{\lambda_1,...,\lambda_n\}$, $\max_j(f_j(s)) =
\max\{f_1(s),f_2(s),...,f_n(s)\}$ and $\min_j(f_j(s)) =
\min\{f_1(s),f_2(s),...,f_n(s)\}$. We note that  if $\mu_{min}>0$, then $\mu_j>0$ for
each $j\in\{1,...,n\}$. Hence, if $n=1$, then $\mu_{min} = \mu_1$. In
particular, this condition is satisfied if each species is selected so that
$\lambda_j \le \bar{s}$ and the growth rate of each species remains positive
throughout each cycle.   
\begin{proposition}\label{prop:constants} Assume the  initial conditions  of 
system
\eqref{competition}  satisfy
$$
s(0)=\bar s^+, \ x_j(0)\geq 0, \  j\in\{1,\ldots,n\},\  \sum_{j=1}^nx_j(0)\neq 
0
$$
and that $\mu_{min}>0$.  Then
all  solutions
remain nonnegative and bounded. If $x_j(0)>0$ for some $j\in \{1,\ldots,n\},$
then $x_j(t)>0,$ for all $t>0.$  
Furthermore, there exists an infinite sequence of times $\{t_k\}_{k\in\mathbb{N}}$ such that $s(t_k^-) = \bar s$ and $t_k \to \infty$ as $k\to\infty.$
\end{proposition}
The conditions of Proposition \ref{prop:constants} ensure that each species is capable of surviving in the fermentor on their own and that the fermentor will cycle indefinitely. In the case where only a single species is present initially (i.e., $x_\ell(0) >0$ for some $\ell\in\{1,...,n\}$ and $x_j(0) = 0 $ if $j\ne \ell$), then model \eqref{competition} reduces to the model studied in~\cite{SmWo2}. We summarize the main results of that paper in the following proposition.
\begin{proposition}[Smith \& Wolkowicz~\cite{SmWo2}]\label{growth1} 
 Fix $\ell\in \{1,2,\dots,n\}.$
 Assume that the  initial conditions  of system \eqref{competition}
satisfy 
$$
s(0)=\bar s^+,  \  x_j(0)=0 \ \text{for}\  j\in \{1,\ldots,n\},   \ j \neq \ell, \
x_\ell(0)>0,
$$
and that $\mu_\ell>0$.
\begin{enumerate}
\item
There exists a unique nontrivial  periodic orbit. This periodic 
orbit
has  exactly one impulse per period and is globally asymptotically 
stable.
\item
At the times of impulse $\{t_k\}_{k\in\mathbb{N}}$, the periodic orbit 
satisfies
\begin{align*}
s(t_k^-)&=\bar s, & s(t_k^+)&=\bar{s}^+, \\
x_\ell(t_k^-) &= \frac{1}{r} \mu_\ell, &
x_\ell(t_k^+)&=
\frac{1-r}{r} \mu_\ell.
\end{align*}
\end{enumerate}
\end{proposition}

\section{Two-species competition in the self-cycling fermentation
process}\label{twospecies}

In this section, we consider pairwise competition between different species.
We assume $\mu_{min}>0$ so that each species is capable of surviving in the fermentor if other species are not present. In the event that one of the species is a strictly better competitor than another species, then the worst competitor will be driven to extinction. 
\begin{proposition}\label{prop:domination} Consider system \eqref{competition} and fix $j,k\in\{1,...,n\}$ with $j\ne k$. If  $f_j(s) - d_j > f_k(s) - d_k$ for all  $s\in(\bar{s},\bar{s}^+)$, then $x_k\to 0$ as $t\to\infty$. 
\end{proposition}
Geometrically, this means that the two response functions must cross at some point in order for coexistence to be possible between these two species. 

We now restrict our attention to model \eqref{competition} in the case where $n=2$. By
Proposition \ref{growth1}, the $(s,x_1,0)$ subspace and $(s,0,x_2)$ subspace each contain a periodic orbit that is globally 
attracting
with respect to solutions with initial conditions in the interior of that subspace.
At the impulse points, these periodic orbits satisfy 
\begin{subequations}
	\begin{align}
(s(t_n^-),x_1(t_n^-),x_2(t_n^-)) &= \left(\bar s, \frac{\mu_1}{r},0\right)
		 \label{eq:po1a} \\
	(s(t_n^+),x_1(t_n^+),x_2(t_n^+)) &= \left(\bar{s}^+,\frac{(1-r)\mu_1}{
		r},0\right) \label{eq:po1b}
\end{align}
\end{subequations}
and
\begin{subequations}
\begin{align} 
(s(t_n^-),x_1(t_n^-),x_2(t_n^-)) &= \left(\bar s,0,\frac{\mu_2}{r}\right) 
	\label{eq:po2a}  \\
(s(t_n^+),x_1(t_n^+),x_2(t_n^+)) &=
\left(\bar{s}^+,0,\frac{(1-r) \mu_2}{ r}
	\right), \label{eq:po2b}
\end{align}
\end{subequations}
respectively.

We analyse the stability of these planar periodic orbits with 
respect to the interior of $\mathbb{R}_+^3$ using impulsive Floquet theory (see~\cite{BS1,BS2}). Each of these periodic orbits has three Floquet 
multipliers;
one of the multipliers equals one, and from
calculations in~\cite{SmWo2}, another multiplier is $1-r$, which is strictly less than one. We  denote the third multiplier for the orbit with $x_j(t)>0$
by
$\Lambda_{jk}$ for $j,k\in \{1,2\}$ with $j\neq k$.

\begin{theorem}\label{thm:persistence} Consider system \eqref{competition} with $n=2$. Assume
$\mu_{min}>0$ and that $|\Lambda_{jk}|>1$ for 
$j,k\in\{1,2\}$ with $j\neq k$. Then all solutions with initial
conditions that satisfy $$
s(0) = \bar s^+, \quad x_1(0) > 0, \quad x_2(0) > 0
$$
are persistent; i.e.,
$$
\liminf_{t\rightarrow\infty}x_1(t) > 0 \quad \text{and} \quad
\liminf_{t\rightarrow\infty}x_2(t) > 0.
$$
\end{theorem}

\bigskip

\noindent
Theorem \ref{thm:persistence} gives conditions under which there 
is
coexistence of the two species, independent of initial conditions 
(provided
both species are present to begin with).  However, it says nothing about 
the
nature of that coexistence. In the special case with $d_1=d_2=0$, we can show that there is an attracting impulsive periodic orbit with one impulse per period. Numerical simulations in the case where $d_j \ne 0$ also indicate that coexistence is in this form. 

\subsection{Competition with $d_1=d_2=0$}

The species-specific death rates are often assumed to be negligible in applications~\cite{liu2014competitive}. This is a valid approximation when the cycle length is not too long, since bacteria in the fermentor will remain in their exponential growth phase for the duration of a cycle. When $d_j = 0$, all of the consumed nutrient is converted to biomass. Without any mass lost to cell death, the total amount of mass in the fermentor is conserved between impulses. As a result, the total mass present in the fermentor converges to a constant value as the number of impulses increases. 

\begin{lemma}\label{lem:conservation}
    Consider system \eqref{competition} with $d_j=0$ for all $j\in\{1,...,n\}$ and
	assume that  the initial conditions satisfy
    $$
    s(0)=\bar s^+, \ x_j(0)\geq 0, \  j\in\{1,\ldots,n\},\  \sum_{j=1}^nx_j(0)\neq 
    0.
    $$
    Then $s+\sum_{j=1}^n x_j \to s^\In$ as $t \to \infty$.
\end{lemma}

As a consequence of 
Lemma \ref{lem:conservation}, we only need to consider solutions of system \eqref{competition} restricted to the set $\{(s,x_1,...,x_n)\in\mathbb{R}^{1+n}_+~|~s+\sum_{j=1}^n x_j = s^\In\}$. Thus, for $n=2$ we consider the reduced system 
\begin{subequations}\label{2comp}
    \begin{align}
        \begin{drcases}
        \frac{dx_1}{dt}  =  x_1 f_1(s^\In-x_1-x_2)\\
        \frac{dx_2}{dt}  =
        x_2 f_2(s^\In - x_1 -x_2)
        \end{drcases} && x_1(t)+x_2(t) &\ne s^\In - \bar{s}, \label{eq:2comp_continuous}\\
        \begin{drcases}
            x_1(t^+) = (1-r)x_1(t^-) \\
            x_2(t^+) = (1-r)x_2(t^-)\qquad\ \,
        \end{drcases}
        && x_1(t^-)+x_2(t^-) &= s^\In - \bar{s},\label{eq:2comp_discrete}
    \end{align}
\end{subequations}
with $(1-r)(s^\In-\bar{s}) = s^\In - \bar{s}^+ \le x_1+x_2\le s^\In-\bar{s}$. If $t_k$ is the $k$th moment of impulse, then we can write $x_2(t_k^+) = s^\In-\bar{s}^+-x_1(t_k^+)$ and $x_2(t_k^-) =s^\In-\bar{s}-x_1(t_k^-)$ by equation \eqref{eq:2comp_discrete}.

\begin{theorem}\label{thm:coexistence}
Consider system \eqref{2comp} with initial conditions satisfying $s^\In-\bar{s}^+\le x_1(0)+x_2(0) < s^\In-\bar{s}  $. Exactly one of the following holds:
\begin{enumerate}
    \item There is at least one periodic orbit with both species present and one impulse per period.
    \item All solutions converge to the periodic orbit
	    \eqref{eq:po1a}--\eqref{eq:po1b} with $x_1(t)>0$ and $x_2(t)=0$.
    \item All solutions converge to the periodic orbit
	    \eqref{eq:po2a}--\eqref{eq:po2b} with $x_1(t)=0$ and $x_2(t)>0$. 
\end{enumerate}
\end{theorem}

Theorem \ref{thm:coexistence} completely characterizes the long-term dynamics of system \eqref{2comp}. Coupling this with Lemma \ref{lem:conservation}, we have a complete understanding of the possible dynamics of system \eqref{competition} when $n=2$ and $d_1=d_2=0$. Thus, if the conditions for Theorem \ref{thm:persistence} are met, then every solution with positive initial conditions must converge to a positive periodic solution with one impulse per period. This discussion suffices as proof of the following corollary.

\begin{corollary}\label{cor:coexistence_d0}
    Consider system \eqref{competition} with $n=2$ and $d_1=d_2=0$. If $\vert\Lambda_{jk}\vert>1$ for $j,k\in\{1,2\}$ with $j\ne k$, then all solutions with initial conditions that satisfy 
    $$ s(0) = \bar{s}^+,\quad x_1(0) >0,\quad x_2(0)>0$$ converge to a positive periodic orbit with one impulse per period. 
\end{corollary}


{\bf Example.} Consider system \eqref{competition} with $n=2$,  $d_1=d_2=0$, and assume the
response functions have Monod form
\begin{eqnarray*}
f_j(s)=\frac{m_j s}{K_j+s}.
\end{eqnarray*}
It can be shown that the Floquet multipliers for the
periodic orbit on the face $x_k\equiv 0$ are 1, $1-r$ and
\begin{eqnarray}
\Lambda_{jk} = \left(\frac{1}{1-r}\right)^{\frac{m_k(K_j+s^\In)}{
m_j(K_k+s^\In)}-1} \left(\frac{K_k+\bar s^+}{
K_k+\bar s}\right)^{\frac{m_k(K_j-K_k)}{ m_j(K_k+s^\In)}}, \qquad j\neq 
k.
\label{muxjk} \end{eqnarray}  
See Appendix~\ref{app:FM} for the calculations of this multiplier. By Corollary
\ref{cor:coexistence_d0}, if $\Lambda_{12}>1$ and $\Lambda_{21}>1$, then solutions converge to a positive periodic solution.  

In Figure \ref{imp01}, we fix the parameters inherent to
the system as well as $m_2$ and $K_2$. This is equivalent to having 
species $x_2$ already in the fermentor. We then vary $m_1$ and $K_1$ to
simulate different possible choices of species $x_1$. Figure \ref{imp01}
shows the various states in $m_1$-$K_1$ space. The other constants are $m_2=1$, $K_2=1$, $s^\In=20$, $\bar s=0.1$ and $r=\frac{1}{2}$.
Two species can coexist in the central shaded region (C). The point $(1,1)$ corresponds to the case when the two uptake functions 
are
identical and both multipliers are equal to one. The bounding curves of 
the green
shaded region (C) are tangent to one another at this point~\cite{Smith?thesis}.

\begin{figure}[ht!]
\includegraphics[width = 0.8\textwidth]{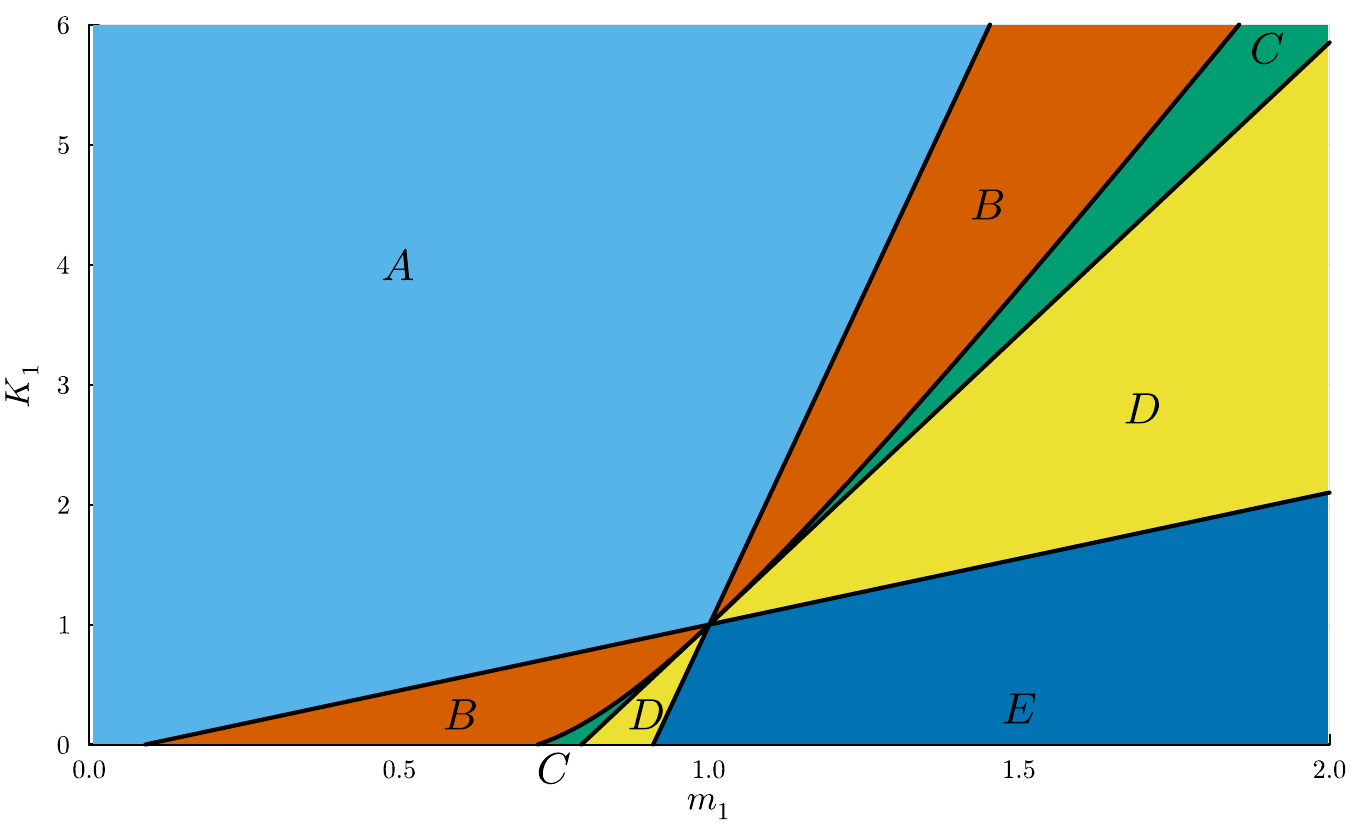}
\caption[Regions in parameter space suggesting coexistence
	]{Outcomes of two species, $x_1$ and $x_2$, competing in the
	fermentor. Parameters $s^\In$, $\bar s$, $m_2$ and $K_2$ were fixed and parameters
	$m_1$ and $K_1$ were varied.
A. $f_2(s)>f_1(s)$ for $\bar s < s < \bar s^+$, so species $x_2$ wins the
competition. B. The two uptake functions cross, but $x_2$ still wins the
competition. C. In the green region in the centre, both Floquet 
multipliers
	are greater than 1, so, as predicted by Corollary~\ref{cor:coexistence_d0}, both
	species coexist. D. The
uptake functions cross, but $x_1$ wins the competition. E. $f_1(s)>f_2(s)$
for $\bar s < s < \bar s^+$, so $x_1$ wins the competition. 
}\label{imp01}
\end{figure}

\section{Three-species competition and simulations} 
\label{threespecies} 

The possibility of survival for two competing species in the self-cycling 
fermentation process raises the question of whether more species can 
coexist on a single nonreproducing limiting nutrient. The 
results in the previous sections cannot easily be applied to competition of $n\ge 3$
species. The impulsive Floquet multipliers can only be calculated
with relative ease for systems that can be reduced to two-dimensional
systems. However, numerical simulations were run to determine whether 
three
species could coexist. 

For the system with three competitors, let $\Lambda_{jk}$ denote the 
nontrivial Floquet multiplier for the periodic orbit on the boundary 
$x_k=0$, for the system where species $j$ and $k$ are present, but the 
third species is absent. Then $\Lambda_{kj}$ is the nontrivial Floquet 
multiplier for the periodic orbit on the boundary $x_j=0$, where the third 
species is absent. We can then apply Theorem \ref{thm:persistence} to each of 
the three cases where two species are present and the third species is 
absent.

System \eqref{competition} with $n=3$ was simulated using the DifferentialEquations.jl toolbox in Julia~\cite{rackauckas2017differentialequations} with $s^\In=20$, $\bar s=0.1$, $r=\frac{1}{2}$ and species-specific parameters listed in Table \ref{tab:parameters_3species}.

\begin{table}[ht]
    \centering
    \begin{tabular}{|c|c|c|c|}
    \hline        $j$ &$m_j$ & $K_j$ & $d_j$  \\
    \hline
        1 & 2.142653 & 6.33 & 0.0 \\
        2 & 1.0 & 1.0 & 0.0 \\
        3 & 7.0 & 32.5 & 0.0\\
        \hline
    \end{tabular}
    \caption{Species-specific parameters used in Figure \ref{comp4graph}. The parameters for Species $x_1$ were chosen from Region D in Figure \ref{imp01}, while the parameters for Species $x_3$ were chosen from Region C.}
    \label{tab:parameters_3species}
\end{table}

Using these data, if $x_1$ is absent, we have 
$$ 
\Lambda_{23} \ = \ 1.137600, \qquad \Lambda_{32} \ = \ 1.049998. 
$$ 
Thus, in the absence of $x_1$, we see that $x_2$ and $x_3$ persist by Theorem 
\ref{thm:persistence}. If $x_2$ is absent, we have 
$$ 
\Lambda_{13} \ = \ 1.008808, \qquad \Lambda_{31} \ = \ 1.014487. 
$$ 
Thus, in the absence of $x_2$, we see that $x_1$ and $x_3$ persist by Theorem 
\ref{thm:persistence}. If $x_3$ is absent, we have 
$$ 
\Lambda_{12} \ = \ 0.985852, \qquad \Lambda_{21} \ = \ 1.089587. 
$$ 
Thus, in the absence of $x_3$, we find that $x_1$ and $x_2$ cannot coexist. It follows 
that
this system is an example of competitor-mediated  coexistence, since $x_2$
cannot survive in the presence of $x_1$ unless $x_3$ is also present. 


\begin{figure}[!tp] 
\includegraphics[]{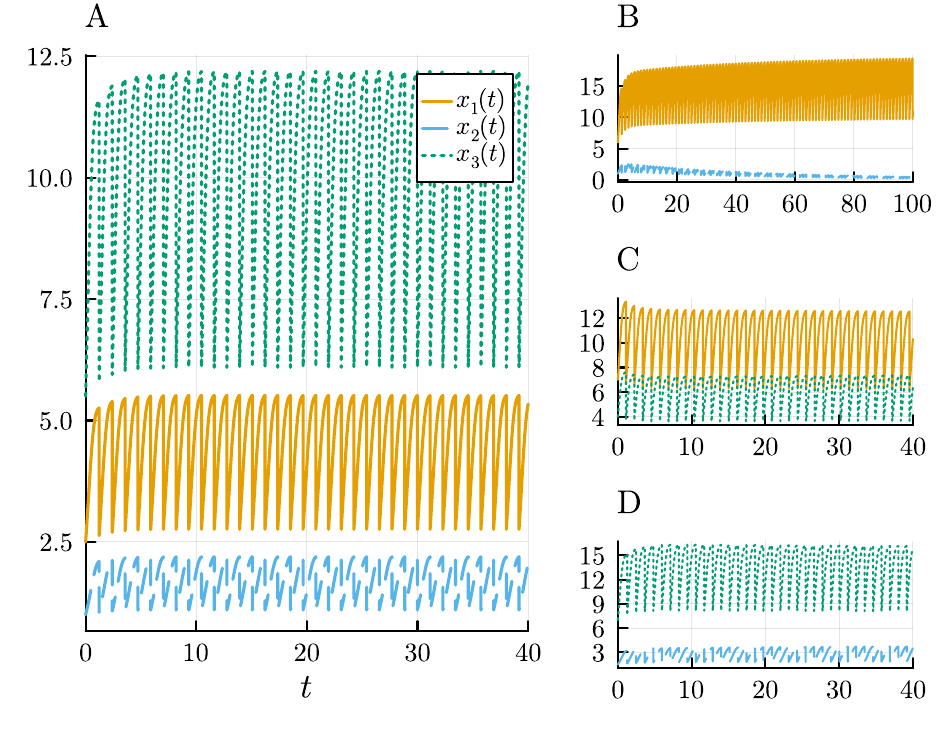}
\caption[Coexistence of three microorganisms in the SCF process]{ 
A. Three microorganisms competing for a single resource in the 
self-cycling fermentor. 
	The species appear to be persisting over time,
suggesting coexistence. Note in particular that $x_2$ is the weakest 
competitor. B. Species $x_2$ cannot survive in the presence of $x_1$ if 
$x_3$ is absent, demonstrating competitor-mediated coexistence. C. Species 
$x_1$ and $x_3$ coexist. D. Species $x_2$ and $x_3$ coexist. 
}\label{comp4graph} 
\end{figure}

\section{Discussion}\label{discussion}

Coexistence of more than one species is possible in the self-cycling
fermentation process. The model with only two species is simple enough that we are
able to prove when two species are able to survive in the same environment using impulsive Floquet theory. However, we are not able to determine the exact form of that coexistence in a general setting. In the special case where the decay rates of both species are negligible, we are able to reduce the dynamics to those of a one-dimensional monotone dynamical system. The general theory of monotone dynamical systems allows us to conclude that coexistence is in the form of a periodic solution with one impulse per period. 

In the analogous model of the chemostat, where the nutrient is pumped in
continuously at a constant rate, coexistence of two species
competing for a single nonreproducing nutrient is not possible (aside from a few knife-edge cases involving the equality of certain parameters)~\cite{SW, wolkowicz1992global}.
The results here are similar to competition in the chemostat with periodic dilution rate~\cite{wolkowicz1998n}. There, multiple species are able to coexist provided that each species is the best competitor for a significant portion of the dilution cycle. A similar condition was required for the coexistence of two species in a model of serial transfer cultures~\cite{smith2011bacterial}. Here, we have extended those results to competition in the self-cycling fermentation process. If one species is the best competitor at every nutrient concentration, then that species will out-compete the others. However, while being a better competitor at some nutrient levels is necessary for survival, it is not sufficient.

We were unable to find analogous theoretical results to determine the outcome of three-species competition. However, numerical simulations show that coexistence between three species is possible. Interestingly, two of the species in our example are unable to coexist without the third species present. This phenomenon of competitor-mediated coexistence has also been observed in other resource-competition models~\cite{arino2003considerations}. 

In applications where the system is best served by a particular class of microorganism, we give conditions for the exclusion of other competing species or strains. Our results suggest that it may be possible to tune reactor parameters, such as the decanting fraction $r$ and decanting criterion $\bar{s}$, in order to exclude unwanted competitors. This could be an important strategy used to maintain desired populations in wastewater treatment systems~\cite{lopez2009modeling,shen2016enhanced,trinh2021recent}.

For applications in which the goal is to maximize the throughput of the system --- as would be the case in the production of polyhydroxyalkanoates --- having multiple species present may provide a more robust system. The coexistence of multiple species offers a buffer in the event that one species abruptly dies off. It is unclear whether the production efficiency would be increased or decreased by the presence of more species, although experimental evidence suggests that an increase is possible~\cite{kourmentza2015polyhydroxyalkanoates}. The fact that three species can co-exist in the self-cycling fermentor suggests the possibility of multiple species co-existing simultaneously under appropriate conditions. This has implications for more efficient treatment of wastewater and greater yield, with a buffer against unexpected species extinction.



\bibliographystyle{plain}
\bibliography{Bibliography}

\appendix
\section{Proofs}\label{proofs}

\begin{proof}[Proof of Proposition \ref{prop:constants}]\mbox{}
That $s(t)$ remains nonnegative is obvious. The faces of $\mathbb{R}^{1+n}_+$ with $x_j = 0$ are invariant under \eqref{competition_continuous}, and therefore, by the uniqueness of solutions to ODEs, the interior of $\mathbb{R}_+^{n+1}$ is invariant. Since impulses take $x_j(t)$ to $(1-r)x_j(t)$, if $x_j(0)>0$, then $x_j(t)>0$ for all $t$.

Next we show that solutions with the given initial conditions
reach $s(t)=\bar s$ in finite time. Suppose not. Then there exists $s^*>\bar{s}$ such that as  $t\to \infty$, $s(t)\to s^*$ and $x_j(t)\to 0$ for each $j\in\{1,...,n\}$. If
$ u(t) = \sum_{j=1}^n x_j(t)$, then $$\frac{du}{ds} = \frac{\sum_{j=1}^n(f_j(s)-d_j)x_j}{-\sum_{j=1}^n f_j(s)x_j}.$$
Integrating with respect to $s$ gives
$$u(s)-u(\bar{s}^+) =  \int_s^{\bar{s}^+}
	\frac{\sum_{j=1}^n(f_j(\sigma)-d_j)x_j(\sigma)}{\sum_{j=1}^nf_j(\sigma)x_j(\sigma)}d\sigma.$$
	If $s^* \ge  \lambda_{max}$, then the integrand is positive for all $s^* < s
	< \bar{s}^+$. This implies that  
$$ 0 > -u(\bar{s}^+)  =  \int_{s^*}^{\bar{s}^+}
	\frac{\sum_{j=1}^n(f_j(\sigma)-d_j)x_j(\sigma)}{\sum_{j=1}^nf_j(\sigma)x_j(\sigma)}d\sigma
	>0,$$ yielding  a contradiction. If $\overline{s} \le s^* < \lambda_{max},$ then \begin{align*}
    -u(\bar{s}^+) &= \int_{s^*}^{\bar{s}^+}\frac{\sum_{j=1}^n(f_j(\sigma)-d_j)x_j(\sigma)}{\sum_{j=1}^nf_j(\sigma)x_j(\sigma)}d\sigma\\\
    &\ge \int_{s^*}^{\lambda_{max}}\frac{\min_j(f_j(\sigma)-d_j)}{\min_j(f_j(\sigma))}d\sigma+\int_{\lambda_{max}}^{\bar{s}^+}\frac{\min_j(f_j(\sigma)-d_j)}{\max_j(f_j(\sigma))}d\sigma\\
    &\ge \mu_{min}, 
\end{align*} 
where the last inequality follows since the integrand of the first integral is negative on the domain of integration. Therefore $\bar{s}$ is reached in finite time and an impulse occurs.

The  solution  is then reset so that  
$s=(1-r)\bar s+rs^\In>\bar s$ 
(since $s^\In>\bar s$), and the sum of the $x_j$'s remain positive. Therefore,
the original  assumptions on the initial conditions are once again 
satisfied.
Hence, solutions cycle indefinitely.
\end{proof}

\begin{lemma}\label{ctsphi} The function $\varphi=\varphi(t,\tau,\xi)$ that
solves \eqref{competition_continuous} for
$n=2$ with initial
condition $\varphi(\tau,\tau,\xi)=\xi$ is continuous in $(t,\tau,\xi)$. 
\end{lemma}

\begin{proof} We have
\begin{eqnarray*}
\left(\begin{array}{c}s'\\x_1'\\x_2'\end{array}\right) =
\left(\begin{array}{c}-f_1(s)x_1-f_2(s)x_2\\x_1(f_1(s)-d_1)\\x_2(f_2(s)-d_2)\end{array}\right) \
=
\ F(t,w)
\end{eqnarray*}
where $w=(s,x_1,x_2)$. Then
\begin{eqnarray*}
F_w = \left[\begin{array}{ccc}
-f_1'(s)x_1-f_2'(s)x_2&-f_1(s)&-f_2(s)\\ x_1f_1'(s)&f_1(s)-d_1&0\\ 
x_2f_2'(s)&0&f_2(s)-d_2
\end{array}\right]
\end{eqnarray*}
Each function $f_i$, $j=1,2$ is continuously differentiable, so
$F$ and $F_w$ are continuous.  Hence $\phi(t,\tau,\xi)$ is continuous
in $(t,\tau,\xi)$ by Theorem 7.1 in~\cite{MM}. \end{proof}

\begin{proof}[Proof of Theorem \ref{thm:persistence}] Consider any initial point 
$(\bar s^+,x_1(0),x_2(0))$ where $x_1(0) > 0$ and $x_2(0) > 0$.
By Proposition \ref{prop:constants}, there exists a first time $t_1$ such that
solutions of \eqref{competition} satisfy
$$
s(t_1^-) \ = \ \bar s, \quad x_1(t_1^-) \ > \ 0, \quad x_2(t_1^-) \ > \ 0.
$$
Then if $t_n$ denotes the time of the $n$th impulse point, we have, for
$t_{n-1}<t < t_n$ ($n> 1$),
\begin{align*}
0< \bar s < s(t) <\bar{s}^+
\end{align*} and,
for
each $j\in\{1,2\}$, either
$$0<(1-r)x_j(t_{n-1}^-)< x_j(t)\quad \text{or}\quad 0<x_j(t_n^-)<x_j(t).$$

Therefore, it suffices to consider the sequence
$\{(u_n,v_n)\}_{n=1}^{\infty}$ where $u_n=x_1(t_n^-)$ and
$v_n=x_2(t_n^-)$  and to show that
$$
\liminf_{n\rightarrow\infty}u_n \ > \ 0 \quad \mbox{and} \quad
\liminf_{n\rightarrow\infty}v_n \ > \ 0.
$$

The equilibrium point $(\bar s,0,0)$ of system \eqref{competition} is unstable
with a one-dimensional centre manifold along the $s$-axis and a two-dimensional unstable manifold that intersects the plane 
\begin{eqnarray*}
S_{\bar s}= \left\{(s,x_1,x_2):s=\bar s, x_1\geq 0, x_2\geq 0\right\}
\end{eqnarray*}
along a smooth curve, say $g(x_1,x_2)=0$. This curve in
$\R^2_+$ connecting the boundary points $(\hat x_1,0)$ and $(0,\hat x_2)$, 
where
$\hat x_1>0$ and $\hat x_2>0$, divides the plane $S_{\bar s}$ into a
bounded region and an unbounded region. Without loss of generality, assume
$$
X = \{(x_1,x_2): x_1\geq 0, x_2\geq 0, g(x_1,x_2)\geq 0 \}
$$
denotes the unbounded region.

Define the map $f:X\rightarrow X$ in the following way:
$$
f(u_n,v_n) = (u_{n+1},v_{n+1}),
$$
where we set $s(0)=\bar s^+, \ x_1(0)=u_n^+,
\ x_2(0)=v_n^+$ and then
determine $u_{n+1}=x_1(t_1^-)$ and $v_{n+1}=x_2(t_1^-)$ from system
\eqref{competition}. 

Note that for any initial condition of the form $s(0)=\bar s^+$, 
$x_1(0)>0$
and $x_2(0)>0$, we have $s(t_1^-)=\bar s$, $x_1(t_1^-)>0$ and $x_2(t_1^-)>0$, so
$s(t_1^+)=\bar s^+$, $x_1(t_1^+)>0$ and $x_2(t_1^+)>0$ and then 
$s(t_2^-)=\bar
s$ and $g(x_1(t_2^-),x_2(t_2^-))>0$. It follows that $f(X)\subset X$.

Next we show that $f$ is continuous on $X$ by showing that $f$ is a
composition $p\circ q:X\rightarrow X$ of two continuous functions,
$$
q:X\rightarrow X \qquad \mbox{and} \qquad p:X\rightarrow X.
$$

Define
\begin{align*}
q(x_1,x_2) &= \left((1-r)x_1,(1-r)x_2\right)
\end{align*}
and
\begin{align*}
p(x_1,x_2) &= (u(x_1,x_2),v(x_1,x_2)),
\end{align*}
where $u(x_1,x_2)=x_1(\bar t)$ and $v(x_1,x_2)=x_2(\bar t)$ such that
$(s(t),x_1(t),x_2(t))$,
$0\leq t\leq \bar t$ is the solution of the associated ODE \eqref{competition_continuous} with initial conditions
$$
s(0) = \bar{s}^+, \quad x_1(0) =  x_1, \quad
x_2(0) 
=
x_2 
$$
and $\bar s <s(t)<(1-r)\bar s+rs^\In$ for $0<t<\bar t$ and
$s(\bar t)=\bar s$.

It is clear that $q$ is continuous. That $p$ is continuous follows from
continuous dependence on initial data for ordinary differential equations
(see Lemma \ref{ctsphi}).
   
The map $f$ has two equilibrium points, $P_1=(\bar x_1,0)$ and 
$P_2=(0,\bar
x_2)$, where $\bar x_j>0$, $j=1,2$. $P_1$ and $P_2$ represent single 
species 
survival equilibria of the map  and correspond to the nontrivial periodic
orbits on the $s$-$x_1$ and $s$-$x_2$ planes, respectively, of system
\eqref{competition}. Each $P_j$, $j=1,2$, is clearly an isolated invariant set.

Assume that $(x_1(t_1),x_2(t_1))$ is any point in $X$ that satisfies 
$x_1(t_1)>0$
and $x_2(t_1)>0$. 
Consider the compact positive orbit 
$\left\{x_1(t_n),x_2(t_n)\right\}_{n\in
\Z_+}$ generated by the map $f$. Assume also that 
$$
\liminf_{n\rightarrow\infty}x_1(t_n) = 0 \quad \mbox{or} \quad
\liminf_{n\rightarrow\infty}x_2(t_n) = 0.
$$
Then either
\begin{enumerate}
\item[(a)] there is a subsequence such that
$$
\lim_{k\rightarrow\infty}x_1(t_{n_k}) = 0 \quad \mbox{and} \quad
\lim_{k\rightarrow\infty}x_2(t_{n_k}) > 0, \qquad \mbox{or}
$$
\item[(b)] there is a subsequence such that
$$
\lim_{k\rightarrow\infty}x_1(t_{n_k}) > 0 \quad \mbox{and} \quad
\lim_{k\rightarrow\infty}x_2(t_{n_k}) = 0.
$$
\end{enumerate}

In case (a), we must have
\begin{eqnarray*}
P_2 \in \omega\left(\left\{x_1(t_n),x_2(t_n)\right\}_{n\in \Z_+}\right).
\end{eqnarray*}
However, since $\Lambda_{21}>1$, the stable manifold of $P_2$ is the set
\begin{eqnarray*}
W^+(P_2) = \{(x_1,x_2):x_1=0,x_2>0\}.
\end{eqnarray*}
Since $x_1(t_{n_k})>0$ for all $k$,
\begin{eqnarray*}
\left\{x_1(t_n),x_2(t_n)\right\}_{n\in \Z_+} \subseteq
W^+_{\mbox{w}}(P_2)\backslash
W^+(P_2).
\end{eqnarray*}
Hence, by Theorem 3.1 of~\cite{FS}, there exists a positive orbit
$\left\{a(t_n),b(t_n)\right\}_{n\in \Z_+}$ in
$$
\omega\left(\left\{x_1(t_n),x_2(t_n)\right\}_{n\in \Z_+}\right)
$$ 
such that
$(a(t_1),b(t_1))\neq P_2$ and 
$$
\left\{a(t_n),b(t_n)\right\}_{n\in \Z_+}\subseteq W^+(P_2).
$$
Hence $a(t_n)=0$ for all $n$. It follows that the omega limit set of
$\left\{x_1(t_n),x_2(t_n)\right\}_{n\in \Z_+}$ is a subset of $W^+(P_2)$.

The orbit $\left\{x_1(t_n),x_2(t_n)\right\}_{n\in \Z_+}$  is a
pseudo-asymptotic orbit of $f$, so by Lemma 2.3 in~\cite{HSZ} the omega limit 
set
is nonempty, compact and invariant. This set cannot include the portion 
of
the  $x_2$
axis above $P_2$, since it is unbounded. 

Consider the set
$$
M = \left\{(0,x_2):\hat x_2\leq x_2\leq \bar x_2\right\}.
$$
Clearly $f(M)\subset M$, but $M\not\subset f(M)$, since $f$ is a
non-decreasing map on $M$ and $f(0,\hat x_2)=(0,y)$ where $y>\hat x_2$. 
Thus
$M$ is not an invariant set.

The only other invariant set in $W^+(P_2)$ is $P_2$
itself. Thus
\begin{eqnarray*}
\omega\left(\left\{x_1(t_n),x_2(t_n)\right\}_{n\in \Z_+}\right)=P_2.
\end{eqnarray*}
However, this implies that 
$$
\left\{x_1(t_n),x_2(t_n)\right\}_{n\in
\Z_+} \subset W^+(P_2),
$$
which is a contradiction.
Thus case (a) is impossible.

Case (b) can be ruled out in a similar fashion.

Hence, for any point $(x_1(t_1),x_2(t_1))$ with $x_1(t_1)>0$, $x_2(t_1)>0$,
we have
$$
\liminf_{n\rightarrow\infty}x_1(t_n)>0 \qquad \mbox{and} \qquad
\liminf_{n\rightarrow\infty}x_2(t_n)>0.
$$\end{proof}

\begin{proof}[Proof of Lemma \ref{lem:conservation}]
Assume that $d_j=0, \ i=1,...n.$  Then, adding together all the equations in 
\eqref{competition}, it follows that between impulses 
$\left(s+\sum_{j=1}^n x_j\right)'(t)=0.$ 
Therefore, for each $k\in\mathbb{N}$, we can define a constant $c_k$ such that 
\begin{align*}
s(t)+\sum_{j=1}^n x_j(t) &=
c_k\end{align*} for $t_k < t < t_{k+1}.$ At the moments of impulse, we have 

\begin{align*}
c_{k+1} &= s(t_{k+1}^+) +\sum_{j=1}^n x_j(t_{k+1}^+)\\
&= rs^\In +(1-r)s(t_{k+1}^-)+(1-r)\sum_{j=1}^n x_j(t_{k+1}^-)\\ 
&= rs^\In + (1-r)c_k,
\end{align*}
a recurrence relation that  has the  general solution
\begin{align*}
c_k &=(1-r)^{k} c_1+  r s^\In ( 1 + (1-r) + (1-r)^2 + \cdots + 
(1-r)^{k-1}),   
\\
  &= (1-r)^{k}c_1 + s^\In (1-(1-r)^{k}), \ \ k\in \mathbb{N}. 
\end{align*}
Therefore,
$ \ \lim_{k\to \infty} c_k = s^\In, \  $
so  it follows that  
$ \ s(t)+\sum_{j=1}^nx_j(t)\to s^\In \ $
as $ \ t \rightarrow \infty.$\end{proof}

\begin{proof}[Proof of Proposition \ref{prop:domination}]
Assume  without loss of generality that $j=1,k=2$ and that $x_1(0)>0$, $x_2(0)>0$. By Proposition \ref{prop:constants}, $x_1(t)>0$ and $x_2(t)>0$ for all $t$, and there is an infinite sequence of impulse times $\{t_\ell\}_{\ell\in\mathbb{N}}$. Thus, the ratio $\frac{x_2(t)}{x_1(t)}$ is well defined. At the moments of impulse, we have $$\frac{x_2(t_\ell^+)}{x_1(t_\ell^+)} = \frac{(1-r)x_2(t_\ell^-)}{(1-r)x_1(t_\ell^-)} = \frac{x_2(t_\ell^-)}{x_1(t_\ell^-)}$$
by equation \eqref{competition}. For $t\in(t_\ell,t_{\ell+1})$ we have $x_i(t) = x_i(t_\ell^+)e^{\int_{t_\ell}^t f_i(s(\xi))-d_i d\xi}$ for $i\in\{1,2\}$ and therefore 
$$ \frac{x_2(t_{\ell+1}^+)}{x_1(t_{\ell+1}^+)} = \frac{x_2(t_\ell^+)}{x_1(t_\ell^+)}e^{\int_{t_\ell}^{t_{\ell+1}} (f_2(s(t))-d_2)-(f_1(s(t))-d_1) dt}.$$
Since $f_1(s)-d_1 > f_2(s)-d_2$ for all $s\in(\bar{s},\bar{s}^+)$, the exponential factor is strictly less than 1. Thus, $$\frac{x_2(t_\ell^+)}{x_1(t_\ell^+)} \to 0 $$ as $\ell\to \infty$. 
\end{proof}

\begin{proof}[Proof of Theorem \ref{thm:coexistence}]
Let $\Gamma^+ = \{(x_1,x_2)\in \mathbb{R}^2_+~|~x_2 = s^\In-\bar{s}^+-x_1\}$ and $\Gamma^- = \{(x_1,x_2)\in \mathbb{R}^2_+~|~x_2 = s^\In-\bar{s}-x_1\}$. Let $\varphi:\Gamma^+\to\Gamma^-$ be the map that takes points in $\Gamma^+$ to points in $\Gamma^-$ along the flow generated by  \eqref{eq:2comp_continuous}. Define
\begin{equation}
    G(x_1) = (1-r)\left(\varphi(x_1,s^\In-\bar{s}^+-x_1)\right)_1,
\end{equation}
where $(\varphi(x_1,x_2))_1$ is the first component of $\varphi(x_1,x_2)$. Fixed points of $G$ correspond to periodic orbits with one impulse per period of system \eqref{2comp}. Note that $x_1= 0$ and $x_1 = s^\In-\bar{s}^+$ are fixed points that correspond to the periodic orbits with only $x_2$ present and only $x_1$ present, respectively. 

The dynamical system defined by iterating $G$ is a one-dimensional monotone dynamical system; by Theorem 5.6 in~\cite{hirsch2006monotone}, every orbit of this dynamical system converges to a fixed point. 

Thus, if there exists $x^*\in(0,s^\In-\bar{s}^+)$ such that $G(x^*)=x^*$, then the solution to system \eqref{2comp} with $(x_1(0),x_2(0)) = (x^*,s^\In-\bar{s}^+-x^*)$ is periodic with one impulse per period. If no such $x^*$ exists, then either $G(x)>x$ or $G(x)<x$ for all $x\in(0,s^\In-\bar{s})$. In the first case, $x_1(t_k^+)$ is increasing with $k$, and all solutions converge to the periodic orbit with $x_2=0$. In the second case, $x_1(t_k^+)$ is decreasing with $k$, and all solutions converge to the periodic orbit with $x_1 = 0$. 
\end{proof}

\section{Floquet Multipliers}\label{app:FM}

Consider the two-dimensional system
\begin{equation}
\begin{aligned}
\frac{ds}{dt} &= P(s,x), & \frac{dx}{dt} &= Q(s,x) &\qquad
(s,x) &\not\in M \\ 
\Delta s &= a(s,x), & \Delta x &= b(s,x) &\qquad (s,x) &\in M,
\end{aligned}\label{autonomy}
\end{equation}
where $t\in \R$, and $M\subset \R^2$ is the set defined by the equation
$\phi(s,x)=0$. 

Assume that  (\ref{autonomy}) has a $T$-periodic solution $\vec p(t) = 
[\gamma(t), 
\eta(t)]$ with
$$
\left|\frac{d\gamma}{dt}\right| +\left|\frac{d\eta}{
dt}\right| \neq 0.
$$
Assume further that the periodic solution $\vec p(t)$ has $q$ instants of
impulsive effect in the interval $(0,T)$. 

One of the Floquet multipliers is equal to 1, since we have a
periodic orbit. From Chapter 8 of Bainov and Simeonov~\cite{BS2}, the other is 
calculated
according to the formula 
\begin{eqnarray}
\mu = \prod_{k=1}^q{\Delta_k\exp \left[\int_0^T{\left( \frac{\partial P
}{\partial s}(\gamma(t),\eta(t)) + \frac{\partial Q}{\partial x}
(\gamma(t),\eta(t))\right) dt} \right]}, \qquad \label{mu_2} 
\end{eqnarray}
where
\begin{eqnarray*}
\Delta_k = \frac{P_+\left(\frac{\partial b }{\partial x} \frac{\partial \phi}{
\partial s} - \frac{\partial b}{\partial s} \frac{\partial \phi}{\partial
x} +
\frac{\partial \phi }{\partial s}\right) + Q_+ \left(\frac{\partial a}{\partial s}
\frac{\partial \phi}{\partial x}-\frac{\partial a}{\partial x} \frac{\partial
\phi}{\partial s} +\frac{\partial \phi }{\partial x}\right) }{
P\frac{\partial
\phi }{ \partial s} +Q \frac{\partial \phi }{ \partial
x}}.
\end{eqnarray*}
Here, $P$, $Q$, $\frac{\partial a}{\partial s}$, $\frac{\partial b }{\partial
s}$,
$\frac{\partial a }{ \partial x}$, $\frac{\partial b }{ \partial x}$,
$\frac{\partial
\phi }{ \partial s}$ and $\frac{\partial \phi }{ \partial x}$ are
computed at
the point $(\gamma(t_k),\eta(t_k))$ and 
$P_+=P(\gamma(t_k^+),\eta(t_k^+))$,
$Q_+=Q(\gamma(t_k^+),\eta(t_k^+))$.

Consider the periodic orbit on the $x_1$ face for system \eqref{competition} with $n=2$.
Denote this periodic orbit by $(\zeta(t),\xi(t),0)$. We use the notation 
$$
\zeta_0=\zeta(0^+), \ \zeta_1=\zeta(T), \ \xi_0=\xi(0^+), \ \xi_1=\xi(T).
$$
From the condition of $T$-periodicity, $\zeta_1^+=\zeta_0$ and
$\xi_1^+=\xi_0$. Thus
\begin{align*}
\zeta_0 &= \bar s^+ &\qquad \xi_0 &= (1-r)({s^\In-\bar s}) \\
\zeta_1 &= \bar s & \xi_1 &= ({s^\In-\bar s}).
\end{align*}
In particular,
\begin{eqnarray*}
\xi_1 = \frac{1}{1-r}\xi_0,
\end{eqnarray*}
and we have the relationship 
\begin{eqnarray}
\zeta(t)+\xi(t) = s^\In \label{s^i_relationship}
\end{eqnarray}
by Lemma \ref{lem:conservation}.

We thus have the two-dimensional system 
\begin{equation}
\begin{aligned}
\frac{ds}{dt} &= - x_1f_1(s) - \left(s^\In-s-x_1\right)f_2(s)
&\qquad s &\neq \bar s \\ 
\frac{dx_1}{dt} &= x_1f_1(s) & s &\neq \bar s \\ 
\Delta s &= -r\bar s+rs^\In & s &= \bar s \\
\Delta x_1 &= -rx_1 & s &= \bar s.
\end{aligned}\label{twosx}
\end{equation}

Using impulsive Floquet theory and (\ref{s^i_relationship}), we have
\begin{align*}
P &= -\frac{1}{1-r}\xi_0 f_1(\bar s) &\qquad P_+ &= -\xi_0 f_1(\bar s^+)
\\  
Q &= \frac{1}{1-r}\xi_0 f_1(\bar s) & Q_+ &= \xi_0 f_1(\bar s^+) \\
\frac{\partial b}{\partial x_1} &= -r & \frac{\partial
\phi}{\partial s} &= 1 \\
\frac{\partial b}{\partial s} &= 0 & \frac{\partial \phi}{\partial x_1} 
&=
0 \\
\frac{\partial a}{\partial s} &= 0 & \frac{\partial a}{\partial
x_1} &= 0.
\end{align*}

Thus
\begin{align*}
\Delta_1 &= \frac{-\xi_0f_1(\bar s^+)\left(-r\cdot
1-0\cdot0+1\right)+\xi_0f_1(\bar s^+)\cdot0}{-\frac{1}{ 
1-r}\xi_0f_1(\bar
s)+\frac{1}{1-r}\xi_0f_1(\bar s)\cdot0} \\ 
&= (1-r)^2\frac{f_1(\bar s^+)}{f_1(\bar s)}. 
\end{align*}
Then, using (\ref{s^i_relationship}), we have
\begin{align*}
\int_0^T{\left[\frac{\partial P }{\partial s}
\left(\zeta(t),\xi(t)\right)+\frac{\partial Q}{\partial
x_1}\left(\zeta(t),\xi(t)\right)\right]dt}
&= \int_0^T{\left[-\xi f_1'(\zeta) +
f_2(\zeta) -\left(s^\In-\zeta-\xi\right)f_2'(\zeta)+f_1(\zeta)\right]dt} \\
&= \int_0^T{\left[-\xi f_1'(\zeta) +
f_1(\zeta) +f_2(\zeta)\right]dt} \\
&= \int_{0}^{T}{\left[\frac{f_1'(\zeta)}{f_1(\zeta)}\zeta' +
\frac{\xi' }{\xi} + f_2(\zeta)\right]dt} \\
&= \int_{\bar s^+}^{\bar s}{\frac{f_1'(\zeta)}{
f_1(\zeta)}d\zeta} + \int_{\xi_0}^{\frac{1}{1-r}\xi_0}{\frac{d\xi}{\xi}} +
\int_0^T{f_2(\zeta)dt} \\ 
&= \ln\left(\frac{f_1(\bar s)}{f_1(\bar s^+)}\right) + \ln \frac{1}{1-r} +
\int_0^T{f_2(\zeta)dt}.
\end{align*}

Now
\begin{align*}
\int_0^T{f_2(\zeta)dt} &=  \int_0^T{\frac{f_2(\zeta)}{-\xi
f_1(\zeta)}\zeta'dt} \\
&= -\int_{\bar s^+}^{\bar s}{\frac{f_2(\zeta) }{
f_1(\zeta)(s^\In-\zeta)} d\zeta} \\
&= \int_{\bar s}^{\bar s^+}{\frac{m_2(K_1+\zeta)}{m_1(K_2+\zeta)(s^\In-\zeta)}d\zeta} \\
&= \frac{m_2}{m_1}\int_{\bar s}^{\bar s^+}{\left[\frac{K_1+s^\In}{(K_2+s^\In)(s^\In-\zeta)}+\frac{K_1-K_2}{(K_2+s^\In)(K_2+\zeta)} \right]d\zeta},
\end{align*}
using partial fraction decomposition. Therefore
\begin{align*}
\int_0^T{\!\!f_2(\zeta)dt} &= \left[-\frac{m_2(K_1+s^\In)}{
m_1(K_2+s^\In)}
\ln (s^\In-\zeta) + \frac{m_2(K_1-K_2)}{m_1(K_2+s^\In)} \ln(K_2+\zeta)\right]_{\bar s}^{\bar s^+} \\ 
&= -\frac{m_2(K_1+s^\In)}{m_1(K_2+s^\In)} \ln (1-r) + \frac{m_2(K_1-K_2)}{m_1(K_2+s^\In)}
\ln\left(\frac{K_2+\bar s^+}{K_2+\bar s}\right) \\
&= \frac{m_2(K_1+s^\In)}{m_1(K_2+s^\In)} \ln \frac{1}{1-r} + 
\frac{m_2(K_1-K_2)}{m_1(K_2+s^\In)} \ln\left(\frac{K_2+\bar s^+}{K_2+\bar
s}\right). 
\end{align*}

Denote the second Floquet multiplier for the periodic orbit on the
$x_1$-axis by $\Lambda_{12}$ and the one on the $x_2$-axis by 
$\Lambda_{21}$.
We thus have $$\Lambda_{12}=
(1-r)^2\frac{f_1(\bar s^+)}{f_1(\bar s)}
\cdot \frac{f_1(\bar s)}{f_1(\bar s^+)}
\cdot \frac{1}{1-r}
\cdot \left(\frac{1}{1-r}\right)^{\frac{m_2(K_1+s^\In)}{m_1(K_2+s^\In)}} \cdot
\left(\frac{K_2+\bar s^+}{K_2+\bar
s}\right)^{\frac{m_2(K_1-K_2)}{ m_1(K_2+s^\In)}} 
$$
\begin{eqnarray}
\Lambda_{12} = \left(\frac{1}{1-r}\right)^{\frac{m_2(K_1+s^\In)}{m_1(K_2+s^\In)}-1}
\cdot \left(\frac{K_2+\bar s^+}{K_2+\bar
s}\right)^{\frac{m_2(K_1-K_2)}{ m_1(K_2+s^\In)}}. \label{mu2x}
\end{eqnarray}

By an identical process applied to the orbit $(\zeta(t),0,\nu(t))$, we
have the the symmetric result
\begin{eqnarray}
\Lambda_{21} = \left(\frac{1}{1-r}\right)^{\frac{m_1(K_2+s^\In)}{
m_2(K_1+s^\In)}-1}
\cdot \left(\frac{K_1+\bar s^+}{K_1+\bar
s}\right)^{\frac{m_1(K_2-K_1)}{m_2(K_1+s^\In)}}. \label{mu2y}
\end{eqnarray}

Note that we can calculate these Floquet multipliers only because the system reduces to a two-dimensional one in each case.

\end{document}